\documentclass[
  hidelinks,
  onefignum,
  onetabnum
]{siamart251216}

\usepackage{mathtools}
\usepackage{amssymb}
\usepackage[english]{babel}
\usepackage{bm}
\usepackage{dsfont}
\usepackage{tikz}
\newsiamthm{thm}{Theorem}
\newsiamthm{lem}{Lemma}
\newsiamthm{prop}{Property}
\newsiamthm{prob}{Problem}
\newsiamthm{Corollary}{Corollary}

\newsiamremark{ass}{Assumption}
\newsiamremark{rem}{Remark}
\usepackage{epsfig}

\input{preamble}

\allowdisplaybreaks

\headers{Nonadaptive Learning in Robust Nonlinear Output Regulation}
{S. Wang, M. Guay, and R. D. Braatz}

\begin{document}
\begin{sloppypar}
%% ============================================================
%% Title and author information
%% ============================================================
\title{
Nonadaptive Learning in Robust Nonlinear Output Regulation
\thanks{
This research was supported by the U.S. Food and Drug Administration
under the FDA BAA-22-00123 program, Award Number 75F40122C00200.
}
}

\author{
Shimin Wang
\thanks{
School of Data Science, Lingnan University, Hong Kong,
and Massachusetts Institute of Technology, Cambridge, MA 02139, USA.
Corresponding author.}
\and
Martin Guay
\thanks{
Department of Chemical Engineering, Queen's University,
Kingston, ON K7L 3N6, Canada
(\email{martin.guay@queensu.ca}).
}
\and
Richard D. Braatz
\thanks{
Massachusetts Institute of Technology, Cambridge, MA 02139, USA
(\email{braatz@mit.edu}).
}
}

\maketitle

%% ============================================================
%% Abstract
%% ============================================================
\begin{abstract}
This paper considers robust nonadaptive regulation for general
nonlinear systems in an output-feedback setting with arbitrarily high
relative degree. We develop a nonadaptive design that combines an
input-driven filter and a generic internal model with a recursive
backstepping law, thereby recasting the regulation problem as the robust
input-to-state stabilization of an augmented error system. Unlike
adaptive schemes, the proposed method does not rely on linearly
parameterized regressors and does not require the construction of
Lyapunov functions having merely nonpositive derivatives. Under standard
assumptions on the exosystem, including purely imaginary and simple
eigenvalues, together with a minimum-phase input-to-state stability
condition on the internal dynamics, we establish global asymptotic
regulation and derive explicit, verifiable inequalities for selecting
the design gains. The resulting nonadaptive framework guarantees
convergence of the estimation and tracking errors even when the
controlled-system dynamics are complex or only partially known. The
effectiveness of the theoretical results is demonstrated using a
benchmark controlled Duffing system.
\end{abstract}

%% ============================================================
%% Keywords and MSC classifications
%% ============================================================


\section{Introduction}
Output regulation is an essential issue in the design of control systems \cite{isidori1990output,sutton2002reinforcement,serrani2000global}. 
It aims to have a system track a class of desired signals while rejecting the external disturbance \cite{marconi2008uniform,isidori1990output,huang2004nonlinear}. %duan2021highVI
In the linear case, regulation reduces to a pole-assignment problem because the steady-state tracking error depends linearly on the exogenous signals \cite{francis1976internal}. 
For nonlinear plants, however, the steady-state error is a nonlinear function of the exogenous signals \cite{huang1994robust,isidori2003robust}, and linear feedforward or linear internal-model designs fail in the presence of nonlinearities or parametric uncertainties.
%Accordingly, the desired signals and external disturbance can be lumped together as exogenous signals that are generated by an autonomous differential equation called the exosystem. 
%
%Various formulations of the output regulation problems have been investigated in the past decade, such as linear systems in \cite{francis1976internal} and nonlinear systems in \cite{marconi2008uniform} and \cite{huang2004nonlinear} with or without uncertainties in the exosystem. 
%
%
%Feedforward and feedback control are widely employed generic schematics for addressing output regulation problems.
%
%In terms of feedforward control, \cite{isidori1990output} showed that the output regulation of nonlinear systems can be solved by a feedforward control synthesized from specific solvable nonlinear partial differential equations called nonlinear regulator equations.
%
%The solvability of the nonlinear output regulation using feedforward control strictly relies on the perfect knowledge of the plant and the exosystem dynamics in the absence of uncertainties. 
%
%{\myr }

To address general nonlinear output regulation problems, various internal-model structures have been developed over the past decades. Canonical linear internal models have been effectively used with adaptive methods to handle uncertain linear exosystems \cite{nikiforov1998adaptive,serrani2001semi,huang2004nonlinear,marino2003output} and more recently for disturbance rejection in Euler–Lagrange systems \cite{wang2022leaderless}. 
For nonlinear exosystems generating non-sinusoidal signals without uncertainty, nonlinear internal models were introduced in \cite{byrnes2003limit} and \cite{chen2015stabilization}. 
A significant milestone is the generic internal model proposed in \cite{marconi2008uniform}, which removes structural assumptions on the steady-state input and accommodates both minimum-phase and non-minimum-phase systems. 
A broader overview can be found in \cite{huang2018internal}.

Adaptive internal-model-based methods provide a mechanism to address parametric uncertainties but suffer from structural limitations. 
Specifically, they require Lyapunov functions with non-positive derivatives and rely on explicit regressors determined by the system and exosystem structure \cite{forte2013robust,liu2009parameter}. 
Such designs offer weaker robustness, as highlighted by counterexamples demonstrating that boundedness may fail even under small external inputs \cite{chen2023lasalle}, and they apply primarily to systems with parametric uncertainty in a suitable regression form. 
These challenges restrict the applicability of adaptive approaches in general nonlinear output regulation.

Nonadaptive methods alleviate these difficulties by avoiding explicit parameter adaptation \cite{isidori2012robust,marconi2007output}, but the generic internal-model approach still depends on an explicit nonlinear continuous mapping that characterizes the steady-state input.
This mapping is assumed to exist \cite{kreisselmeier2003nonlinear} but is generally unknown, and no general analytic expression is available.%, as emphasized in \cite{bin2024robust}.
Existing solutions often rely on numerical least-squares techniques to approximate this function \cite{marconi2008uniform,marconi2007output}, which introduces additional tuning requirements and may limit robustness or generality.
Recent work \cite{wang2023nonparametric} addresses
this issue for nonlinear output regulation problems of relative degree one by constructing the
steady-state mapping nonparametrically, without requiring parametric regressors or restrictive
exosystem assumptions. 

%This development broadens the applicability of generic internal-model techniques and provides a foundation for the nonadaptive framework adopted in this article.

%

Motivated by these developments, the present article addresses the nonlinear robust output
regulation problem for general nonlinear output-feedback systems with arbitrary relative degree by employing nonadaptive learning methods. 
This setting is considerably more challenging
than the relative-degree-one case treated in \cite{wang2023nonparametric}, as the controller no longer has access to derivatives of the regulated output and the internal model and stabilization structures must be redesigned accordingly. 
The proposed approach
differs from existing techniques based on adaptive control \cite{tomei2023adaptive,liu2009parameter,
liuzzo2007adaptive,ding2003global}, which typically require explicit regressors, parameter
adaptation laws, and Lyapunov functions with non-positive derivatives, whose applicability is often restricted to parametric uncertainty structures.
Output regulation for uncertain nonlinear systems with higher relative degree remains an active research topic \cite{tomei2023adaptive,
dimanidis2020output}, with recent work focusing on low-complexity, approximation-free output-feedback
designs capable of achieving prescribed transient and steady-state performance \cite{dimanidis2020output}.
In contrast, the nonadaptive learning framework developed in this article addresses the global robust
output regulation problem through the construction of a linear generic internal model based on the
existence of a continuous nonlinear steady-state mapping.
This construction transforms the nonlinear
robust output regulation problem into a robust nonadaptive stabilization problem for an augmented
system with Input-to-State Stable (ISS) dynamics.
%
%Furthermore, integrating a nonparametric learning framework ensures the viability of the nonlinear mapping, demonstrating its capability to capture intricate and nonlinear relationships without the need for the existence of a predefined model of the steady-state input behaviour. 
%

The rest of this paper is organized as follows.  Section \ref{section2} introduces some standard assumptions and lemmas. Section \ref{mainresults} is devoted to the presentation of the main results. This is followed by simulation examples in Section \ref{numerexam} and brief conclusions in Section~\ref{conlu}.

\textbf{Notation:} $\|\cdot\|$ is the Euclidean norm. $\emph{Id} : \mathds{R}\rightarrow \mathds{R}$ is an identity function. For $X_i\in \mathds{R}^{n_i\times m}$ with $i=1,\dots,N$, let $\col(X_1,\dots,X_N)=[X_{1}^{\!\top},\dots ,X^{\!\top}_{N}]^{\!\top}$ and $$\textnormal{diag}(X_1,\dots,X_N)=\left[\begin{matrix}X_1& & \\
& \ddots & \\
& & X_N\end{matrix}\right].$$
A function $\alpha: \mathds{R}_{\geq 0}\rightarrow \mathds{R}_{\geq 0}$ is of class $\mathcal{K}$ if it is continuous, positive definite, and strictly increasing. The notation $\mathcal{K}_{\infty}$ identifies the subclasses of unbounded $\mathcal{K}$ functions. For functions $f_1(\cdot)$ and $f_2(\cdot)$ with compatible dimensions, their composition $f_{1}(f_2(\cdot))$ is denoted by $f_1\circ f_2(\cdot)$. For a matrix $X$,  $\textnormal{\mbox{Adj}}[X]$ denotes its adjugate matrix.

\section{Problem Formulation and Assumptions}\label{section2}

We consider a class of nonlinear control systems of the form:
\begin{subequations}\label{second-nonlinear-systems}
\begin{align}
\dot{z}&=f(z,y,v,w),\label{second-nonlinear-systems-a}\\
\dot{x}&= A_cx+g(z,y,v,w)+B_cbu,\\
 y &=C_cx, \\
  e &= y - h(v,w),
 \end{align}
\end{subequations}
where $(z, x)\in \mathds{R}^{n_z}\times \mathds{R}^{r}$ is the vector of state variables with $r\geq 1$ with a $z$-subsystem of fully nonlinear dynamics and a $x$-subsystem of partially structured linear dynamics with an additional nonlinear term  $g(z,y,v,w)=\col(g_1(z,y,v,w),\dots, g_r(z,y,v,w))$, $y\in \mathds{R}$ is the output of the system, $h(v,w)$is explicitly the \emph{reference output} generated by the
exosystem \eqref{eqn: exosystem system}, $ e\in \mathds{R}$ is the tracking error of the system, $u\in \mathds{R}$ is the input, $w\in \mathds{W}\subset \mathds{R}^{n_w}$ is an uncertain parameter vector with $\mathds{W}$ being an arbitrarily prescribed subset of $\mathds{R}^{n_w}$ containing the origin, $b$ is a positive constant, the functions $h(\cdot)$, $f(\cdot)$, $g_i(\cdot)$ are globally defined and sufficiently smooth and satisfy $f(0,0,0,w)=0$, and $g_i(0,0,0,w)=0$ for all $w\in \mathds{W}$, and $v(t)\in \mathds{R}^{n_v}$ is an exogenous signal representing the reference input and disturbance, which is generated by the exosystem
\begin{align}\label{eqn: exosystem system}
\dot{v}=&S(\sigma)v,
\end{align}
where $\sigma \in \mathds{S}\subset \mathds{R}^{n_\sigma}$ represents the uncertainties in the exosystem with $S(\sigma)$ being a constant matrix. The matrices $A_c\in \mathds{R}^{r\times r}$, $C_c\in \mathds{R}^{1\times r}$, and $B_c\in \mathds{R}^r$ have the form
\begin{align*}
    A_c=&\left[\begin{matrix}\textbf{0}& I_{r-1}\\
    0& \textbf{0}
    \end{matrix}\right]\!, \ \  C_c^{\!\top}=\col(1, \textbf{0}_{r-1}), \ \ B_c=\col(\textbf{0}_{r-1},1),
\end{align*}
where $I_{r-1}$ and $\textbf{0}_{r-1}=\col(0,\dots,0)$ are the identity matrix and zero vector of $r-1$ dimension, respectively.

The nonlinear robust output regulation problem in this article is formulated below
\begin{prob} \label{Prob: second-order-Output-regulation}
Given the nonlinear system \eqref{second-nonlinear-systems}--\eqref{eqn: exosystem system} and any compact subsets $\mathds{S}\in \mathds{R}^{n_{\sigma}}$, $\mathds{W}\in \mathds{R}^{n_w}$, and $\mathds{V}\in \mathds{R}^{n_v}$ with $\mathds{W}$ and $\mathds{V}$ containing the origin, design a control law 
such that for all initial conditions $v(0)\in \mathds{V}$, $\sigma \in \mathds{S}$ and $w\in \mathds{W}$, and any initial states $\textnormal{\col}(z(0), x(0))\in\mathds{R}^{n_z+r}$, the solution of the closed-loop system exists and is bounded for all $t\geq 0$,
 and $\lim\limits_{t\rightarrow\infty}e(t)=0$.
\end{prob}
Before proceeding with the main results, we state the assumptions.
\begin{ass}\label{ass0} For all $\sigma$, all the eigenvalues of $S(\sigma)$ are simple
with zero real parts.
\end{ass}
\begin{ass}\label{ass1}  There exists a globally defined smooth function $\bm{z}(v,w,\sigma):  \mathds{R}^{n_v}\times \mathds{R}^{n_w}\times \mathds{R}^{n_\sigma} \longmapsto \mathds{R}^{n_z}$ such that
\begin{align}\label{regulator-1}
\frac{\partial \bm{z}(v,w,\sigma)}{\partial v} S(\sigma)v=f(\bm{z}(v,w,\sigma), h(v,w), v, w)
\end{align}
for all $(v, w, \sigma)\in \mathds{V}\times\mathds{W} \times\mathds{S}$ with $\bm{z}(0,w,\sigma)=0$.
\end{ass}
\begin{rem} The term $ \bm{z}(v,w,\sigma)$ denotes the steady state of the $z$-subsystem \eqref{second-nonlinear-systems-a}, which is the solution to the associated regulator equation \eqref{regulator-1}. Assumption \ref{ass0} is a standard assumption appearing in \cite{isidori2012robust,huang2004general,serrani2001semi,wang2023nonparametric} which limits the exogenous signal $v$ generated in \eqref{eqn: exosystem system} to be arbitrarily large constant signals and multi-tone sinusoidal signals with arbitrarily unknown initial phases, amplitudes and arbitrarily known frequencies. 
\end{rem}
\begin{ass}[Minimum-phase condition]\label{H2}  The translated inverse system 
\begin{align}\label{Reg-2}
\dot{\bar{z}} &={f(\bar{z}+\bm{z}(\mu),e+h(v,w),v, w)-f(\bm{z}(\mu), h(v,w),v, w) }%_{\bar{f}(\bar{z},e,\mu)}
\end{align}
is input-to-state stable with state $\bar{z}=z-\bm{z}(\mu)$, $\mu=\textnormal{\col}(v,w,\sigma)$ and input $e$ in the sense of \cite{Sontag2019}. In particular,  
there exists a continuous function $V_{\bar{z}}(\bar{z})$ satisfying 
$$\underline{\alpha}_{\bar{z}}(\|\bar{z} \|)\leq V_{\bar{z}}(\bar{z}) \leq \overline{\alpha}_{\bar{z}}(\|\bar{z} \|)$$
for some class $\mathcal{K}_{\infty}$ functions $\underline{\alpha}_{\bar{z}}(\cdot)$ and $\overline{\alpha}_{\bar{z}}(\cdot)$ such that, for any $v\in \mathds{V}$, along the trajectories of the $\bar{z} $ subsystem,
$$\dot{V}_{\bar{z} }\leq-\alpha_{\bar{z}}(\|\bar{z} \|)+ \gamma(e),$$
where $\alpha_{\bar{z}}(\cdot)$ is some known class $\mathcal{K}_{\infty}$ function satisfying $\limsup\limits_{\varsigma\rightarrow 0^{+}}\left(\alpha_{\bar{z}}^{-1}(\varsigma^2)/\varsigma\right)\!< + \infty$, and $\gamma(\cdot)$ is some known smooth positive definite function.
\end{ass}
\begin{rem}\label{rem-barV} Assumption \ref{H2} guarantees that the $\bar{z}$-system \eqref{Reg-2} is input-to-state stable with state $\bar{z}$ and input $e$. In addition, by the changing supply function technique \cite{sontag1995changing}, there exists a continuous function $\bar{V}_{\bar{z}}(\bar{z})$ satisfying 
$ \underline{\alpha}_{\bar{z}}(\|\bar{z}\|)\leq  \bar{V}_{\bar{z}}(\bar{z})\leq \bar{\alpha}_{\bar{z}}(\|\bar{z}\|)$, for some class $\mathcal{K}_{\infty}$ functions $\underline{\alpha}_{\bar{z}}(\cdot)$ and $\bar{\alpha}_{\bar{z}}(\cdot)$ such that  for any smooth function $\Delta_{\bar{z}}(\bar{z})>0$ and any $\mu \in \mathds{V}\times\mathds{W} \times\mathds{S}$, the time derivative of $ \bar{V}_{\bar{z}}(\bar{z})$ along the trajectory \eqref{Reg-2} satisfies the inequality 
$$\dot{\bar{V}}_{\bar{z}}(\bar{z})\leq -\Delta_{\bar{z}}(\bar{z})\|\bar{z}\|^2+\delta_{\bar{z}}\gamma_{\bar{z}}(e)e^2,$$
where $\delta_{\bar{z}}$ and $\gamma_{\bar{z}}(\cdot)$ are some positive constant and positive smooth function, respectively.
%This Assumption \ref{H2} be used for introducing Properties \label{property1} and \label{property2}.
\end{rem}
Under Assumption \ref{ass1}, there exists globally defined smooth functions $\{\bm{x}(v,w,\sigma), \bm{z}(v,w,\sigma), \bm{u}(v,w,\sigma)\}$ with $\bm{x}_1(v,w,\sigma)=h(v,w)$ satisfying 
\begin{align*}
\frac{\partial\bm{x}(\mu)}{\partial v}S(\sigma)v &=  A_c\bm{x}(\mu)+ g(\bm{z}(\mu),\bm{x}_1(\mu),\mu)+B_cb\bm{u}(\mu),\\
    \bm{u}(\mu)&= b^{-1}\!\times \left[\frac{\partial\bm{x}_r(\mu)}{\partial v}S(\sigma)v-g_r(\bm{z},\bm{x}_1(\mu),\mu)\right]\!,
\end{align*}
with $\bm{x}(\mu)=\col(\bm{x}_1(\mu),\dots, \bm{x}_r(\mu))$. 
For convenience, let $\bm{x}\equiv \bm{x}(\mu)$, $\bm{z}\equiv \bm{z}(\mu)$ and $\bm{u}\equiv \bm{u}(\mu)$. 

The control system \eqref{second-nonlinear-systems} has relative degree $r\geq 2$. Motivated by \cite{jiang2004unifying},  we define the input-driven filter 
\begin{align}\label{input-filter}
    \dot{\hat{x}}=A\hat{x}+B_cu,
\end{align}
where $\hat{x}\in \mathds{R}^{r}$ is an estimate of $x$ in system \eqref{second-nonlinear-systems}, and $A=A_c-\lambda C_c$ with $\lambda=\col(\lambda_1,\dots,\lambda_r)$ such that $A$ is Hurwitz. We perform the coordinate transformation $\tilde{x}_i=b^{-1}x_i-\hat{x}_i$, $i=1,\dots,r$, to obtain
\begin{subequations}\label{nonlinear-systems-trans}
\begin{align}
\dot{z}&=f(z,y,v,w),\\
\dot{\tilde{x}}&= A\tilde{x}+b^{-1}\times(\lambda y+g(z,y,v,w)),\\
 \dot{y} &=b\hat{x}_2+b\tilde{x}_2+g_1(z,y,v,w), \\
  \dot{\hat{x}}_i&= \hat{x}_{i+1}-\lambda_i\hat{x}_1,\;\;i=2,\dots,r-1,\;\;\\
  \dot{\hat{x}}_r&=u-\lambda_r\hat{x}_1,
 \end{align}
\end{subequations}
where $g(z,y,v,w)=\col(g_1(z,y,v,w),\dots, g_r(z,y,v,w))$ and $\tilde{x}=\col(\tilde{x}_1,\dots,\tilde{x}_r)$.

\begin{ass}\label{ass0i} The function $\bm{u}(v,\sigma,w)$ is a
polynomial in $v$ with coefficients depending on $w$ and $\sigma$.
\end{ass}
\begin{lem}\label{lem:steady_state_generator}
Suppose Assumption \ref{ass1} holds. If the function $\bm{u}(v,\sigma,w)$ is a polynomial in $v$ with coefficients depending on $w$ and $\sigma$, then there exist a dimension $n_{\tau}$, a vector of monomials $\tau_u(v) \in \mathds{R}^{n_{\tau}}$, and matrices $\Phi_u(\sigma) \in \mathds{R}^{n_{\tau}\times n_{\tau}}$ and $\Gamma_u(\sigma, w) \in \mathds{R}^{1\times n_{\tau}}$ such that:
\begin{align}\label{steady-tau-u}
    \frac{ \partial \tau_u(v)}{\partial v}S(\sigma)v &= \Phi_u(\sigma) \tau_u(v), \\
    \bm{u}(\mu) &= \Gamma_u(\mu) \tau_u(v), \nonumber
\end{align}
where all eigenvalues of $\Phi_u(\sigma)$ have zero real part. Furthermore, if the Hurwitz matrix $A$ shares no common eigenvalues with $\Phi_u(\sigma)$, there exists a unique matrix $P_u \in \mathds{R}^{r\times n_{\tau}}$ such that $\bm{\hat{x}}(\mu) \coloneqq P_u\tau_u(v)$ satisfies the regulator equation:
\begin{equation}\label{regulator-2}
    \dot{\bm{\hat{x}}}(\mu) = A \bm{\hat{x}}(\mu)+ B_c \textbf{u}(\mu).
\end{equation}
\end{lem}

\begin{proof}
Since $\bm{u}(v,\sigma,w)$ is a polynomial in $v$, there exists an integer $n_u$ such that $\bm{u}$ can be expanded as:
\begin{align*}
\bm{u}(v,\sigma,w)=& \sum\limits_{l=1}^{n_u} U_{l}(\sigma,w) v^{[l]}\\
=&\ \underbrace{\left[\begin{matrix}U_{1}(\sigma,w)&\cdots &U_{n_u}(\sigma,w)\end{matrix}\right]}_{\Gamma_u(\sigma,w)} \col(v^{[1]},\dots, v^{[n_u]})
\end{align*}
where $\Gamma_u(\sigma,w)$ is a suitable coefficient matrix and $v^{[l]}$ represents the vector of monomials of degree $l$:
\begin{align*}
v^{[l]}=&\ \col(v_1^{l},v_1^{l-1}v_2,\dots,v_1^{l-1}v_{n_u}, \dots, v_{n_u}^{l}), \quad l= 1,\dots, n_u.
\end{align*}
Let $\tau_u(v)=\col(v^{[1]},\dots, v^{[n_u]})$. Following Chapter 4 of \cite{huang2004nonlinear}, the time derivative of the monomial vector satisfies linear dynamics driven by a matrix $\Phi_u(\sigma)$ with eigenvalues on the imaginary axis, proving \eqref{steady-tau-u}.

Since $A$ is Hurwitz and $\Phi_u(\sigma)$ has eigenvalues with zero real parts, they share no common eigenvalues. Thus, the generalized Sylvester equation:
\begin{equation*}
    P_u\Phi_u(\sigma)=A P_u+B_c \Gamma_u(\mu)
\end{equation*}
admits a unique solution $P_u$. Defining $\bm{\hat{x}}(\mu)=P_u\tau_u(v)$ and taking its time derivative along the trajectories of the exosystem yields:
\begin{align*}
    \dot{\bm{\hat{x}}}(\mu) &= P_u \frac{\partial \tau_u(v)}{\partial v} S(\sigma)v \\
    &= P_u\Phi_u(\sigma) \tau_u(v) \\
    &= (A P_u + B_c \Gamma_u) \tau_u(v) \\
    &= A \bm{\hat{x}}(\mu)+ B_c \textbf{u}(\mu).
\end{align*}
This completes the proof.
\end{proof}

Let $\bm{E}(\mu)=b^{-1}\bm{x}(\mu)-\bm{\hat{x}}(\mu)$; then the regulator equation solution associated with the composite systems \eqref{eqn: exosystem system} and \eqref{nonlinear-systems-trans} is
$$\{\bm{z}(\mu),\ \bm{E}(\mu),\ \bm{y}(\mu),\ \bm{\hat{x}}(\mu),\ \bm{u}(\mu)\}.$$

%Some more standard assumptions are below.
%\begin{ass}\label{ass0i} The function $\bm{\hat{x}}_2(v,\sigma,w)$ is
%polynomial in $v$ with coefficients depending on $w$ and $\sigma$.
%\end{ass}
% \begin{rem}\myr 
    
% \end{rem}
 \begin{rem}\label{remPE} The second element of the vector $\bm{\hat{x}}(\mu)=\textnormal{\col}(\bm{\hat{x}}_1(\mu),\dots, \bm{\hat{x}}_r(\mu))$ with $\mu=\textnormal{\col}(v,\sigma,w)$ is denoted by $\hat{\bm{x}}_2(\mu)$. Let $\bm{\hat{x}}$ be the steady state associated with \eqref{input-filter} and, $\bm{\hat{x}}(\mu)$, the solution to the regulator equation \eqref{regulator-2}.
From \cite{liu2009parameter}, under Assumptions \ref{ass0} and \ref{ass0i}, for the function $\bm{\hat{x}}_2(v,\sigma,w)$, there exist integers $n>0$ such that 
 $\bm{\hat{x}}_2(v,\sigma,w)$ can be expressed as
\begin{align}\label{remPE-trsin}
%\bm{x}_2(v(t),\sigma,w)=&\sum_{j=1}^{n_1^*}C_{1j}(v(0), w,\sigma)e^{\imath \hat{\omega}_{1j}t},\\
\bm{\hat{x}}_2(v(t),\sigma,w)=&\sum\nolimits_{j=1}^{n}C_{j}(v(0), w,\sigma)e^{\imath \hat{\omega}_{j}t}
\end{align}
for some functions $C_{j}(v(0), w,\sigma)\in \mathds{C} $, where $\imath$ is the imaginary unit and $\hat{\omega}_j$ are distinct real numbers for $1\leq j \leq n$. The minimal zeroing polynomial of $\bm{\hat{x}}_2(v(t),\sigma,w)$ is $\Pi_{j=1}^{n}(s+\imath \hat{\omega}_{j} )$.
\end{rem}
\begin{ass}\label{ass5-explicit}  
 The initial condition $v(0)\in \mathds{V}$ and any parameter vectors $w\in \mathds{W}$ and $\sigma\in \mathds{S}$ satisfy the coefficients satisfy $C_{j}(v(0), w,\sigma)\neq 0$ for all $1\leq j \leq n$.
\end{ass}

\begin{rem}[On Assumption~\ref{ass5-explicit}]
Assumption~\ref{ass5-explicit} imposes a nondegeneracy condition on the coefficients 
$C_j(v(0),w,\sigma)$ appearing in the steady-state generator. Its purpose is to ensure that the
nonlinear mapping used in the construction of the generic internal model is well defined and does 
not encounter singularities for any admissible exosystem initial condition or parameter values.
This prevents the generalized Sylvester-type equation from losing rank and guarantees the existence
and uniqueness of the steady-state input required by the internal model. Similar nondegeneracy 
assumptions appear in classical nonlinear output regulation theory, e.g., in the solvability 
conditions of the nonlinear regulator equations and steady-state mappings in 
\cite{isidori1990output,huang2004nonlinear} and \cite{marconi2008uniform}, as well as in our previous 
nonparametric framework for relative-degree-one systems \cite{wang2023nonparametric}. 
Assumption~\ref{ass5-explicit} should therefore be understood as a mild structural condition that 
avoids singularities rather than a stability or ISS requirement.
\end{rem}

\subsection{Generic internal model design}
Under Assumptions \ref{ass0} and \ref{ass0i}, there exists a positive integer $n$ such that $\bm{\hat{x}}_2(\mu)$ satisfies,  for all $\mu\in \mathds{V}\times\mathds{W}\times \mathds{S}$,
%\begin{subequations}
\begin{align} \label{aode-explicit}
%\frac{d^{n_1}\bm{x}_2(v,\sigma,w)}{dt^{n_2}}+a_{1}&(\sigma)\bm{x}_2(v,\sigma,w)+a_{2}(\sigma)\frac{d\bm{x}_2(v,\sigma,w)}{dt}\notag\\
%+&\dots+a_{n_1}(\sigma)\frac{d^{n_2-1}\bm{x}_2(v,\sigma,w)}{dt^{n_1-1}}=0,\\
\frac{d^{n}\bm{\hat{x}}_2(\mu)}{dt^{n}}&+a_{1}(\sigma)\bm{\hat{x}}_2(\mu)
+\dots+a_{n}(\sigma)\frac{d^{n-1}\bm{\hat{x}}_2(\mu)}{dt^{n-1}}=0,
\end{align}
%\end{subequations}
where $a_{1}(\sigma),\dots,a_{n}(\sigma)$ belong to $\mathds{R}$. Under Assumptions \ref{ass0} and \ref{ass0i}, equation \eqref{aode-explicit} yields the polynomial $$\varsigma^{n}+a_{1}(\sigma)+a_{2}(\sigma)\varsigma+\dots+a_{n}(\sigma)\varsigma^{n-1}$$ whose roots are distinct with zero real parts for all $\sigma\in \mathds{S}$. Let $a(\sigma)=\col(a_{1}(\sigma), \dots, a_{n}(\sigma))$, $\bm{\xi}(\mu)= \col\!\left(\bm{\hat{x}}_2(\mu),\frac{d\bm{\hat{x}}_2(\mu)}{dt},\dots,\frac{d^{n-1}\bm{\hat{x}}_2(\mu)}{dt^{n-1}}\right)\!$, and $\bm{\xi} \equiv \bm{\xi}(\mu)$, and define
\begin{align*}
  \Phi(a(\sigma)) =&\left[
                      \begin{array}{c|c}
                        \textbf{0}_{(n-1)\times 1} & I_{n-1} \\
                        \hline
                        -a_{1}(\sigma) &-a_{2}(\sigma),\dots,-a_{n}(\sigma) \\
                      \end{array}
                    \right], \\
  \Gamma =&\left[
                \begin{array}{cccc}
                  1 & 0 & \cdots &0 \\
                \end{array}
              \right]_{1\times n}.
\end{align*}
Then, $ \bm{\xi}\left(\mu\right)$, $\Phi(a(\sigma))$ and $\Gamma$ satisfy
\begin{subequations}\label{stagerator}
\begin{align}
\dot{\bm{\xi}}(\mu)&=\Phi (a(\sigma))  \bm{\xi}(\mu),\\
\bm{\hat{x}}_2(\mu)
   &= \Gamma \bm{\xi}(\mu).
\end{align}
\end{subequations}
System \eqref{stagerator} is called a steady-state generator with output $\hat{x}_2$ as it can be used to produce the steady-state signal $\bm{\hat{x}}_2$. Define the matrix pair $(M, N)$ by
\begin{subequations}\label{MNINter}\begin{align}
M=&\left[
                      \begin{array}{c|c}
                        \textbf{0}_{(2n-1)\times 1} & I_{2n-1} \\
                        \hline
                        -m_{1} &-m_{2},\dots,-m_{2n} \\
                      \end{array}
                    \right],\\
N=&\left[
                \begin{array}{cccc}
                  0 & \cdots & 0 &1 \\
                \end{array}
              \right]_{1\times 2n}^{\!\top},                   
\end{align}
\end{subequations}
where $m_{1}$, $m_{2}$, $\dots$, $m_{2n}$ are chosen such that $M$ is Hurwitz, together with all the eigenvalues of $\Phi(a)$ being distinct with zero real parts, which results in  
the nonsingular matrix-valued function
$$ \Xi(a) \equiv \Phi(a)^{2n}+\sum\nolimits_{j=1}^{2n}m_{j}\Phi(a)^{j-1} \in \mathds{R}^{n \times n }.$$ Then, using $\Xi(a)\Phi(a)=\Phi(a)\Xi(a)$ and $\col(\Gamma, \Gamma \Phi(a),\dots,  \Gamma\Phi(a)^{n-1} )=I_n$ gives that $\Phi(a)\Xi(a)^{-1}=\Xi(a)^{-1}\Phi(a)$ and
\begin{align}
&\ \col(Q_1(a),Q_2(a),\dots,Q_n(a))\nonumber\\
=&\ \col(\Gamma \Xi(a)^{-1},\Gamma \Xi(a)^{-1}\Phi(a),\dots, \Gamma \Xi(a)^{-1}\Phi(a)^{n-1})\nonumber\\
=&\ \col(\Gamma \Xi(a)^{-1},\Gamma \Phi(a)\Xi(a)^{-1},\dots, \Gamma \Phi(a)^{n-1}\Xi(a)^{-1})\nonumber\\
=&\ \underbrace{\col(\Gamma,\Gamma \Phi(a),\dots,  \Gamma \Phi(a)^{n-1})}_{I_n} \Xi(a)^{-1} \label{XIQA-explicit}
\end{align}
with  $Q_{j}(a)=\Gamma \Xi(a)^{-1}\Phi(a)^{j-1} \in \mathds{R}^{1\times n}$, $j=1,\dots, n$.
Define the Hankel real matrix \cite{afri2016state}: 
\begin{align}
\Theta (\theta)&\equiv\!\left[\begin{matrix}\theta_{1} &\theta_{2}&\cdots&\theta_{n}\\
\theta_{2}&\theta_{3}&\cdots&\theta_{n+1}\\
\vdots&\vdots&\ddots&\vdots\\
\theta_{n} &\theta_{n +1}&\cdots&\theta_{2n -1}
\end{matrix}\right]\!\in \mathds{R}^{n \times n},\nonumber
\end{align}
where $\theta= \textnormal{\col}(\theta_{1}, \theta_{2}, \dots, \theta_{2n})= Q \bm{\xi} $  with 
\begin{align}\label{Qdefini}
Q\equiv\textnormal{\col}(Q_{1},\dots,Q_{2n})\in \mathds{R}^{2n\times n},
\end{align}
and \begin{align*}Q_{j}(a)=\Gamma \Xi(a)^{-1}\Phi(a)^{j-1} \in \mathds{R}^{1\times n},& &1\leq j\leq 2n. \end{align*}
Under Assumptions \ref{ass0}, \ref{ass1} and \ref{ass0i}, the matrices $\Phi(a(\sigma))$, $Q$, $M$, $N$, and $\Gamma$ satisfy the matrix equation (see \cite{wang2023nonparametric}):
 \begin{align}
M Q &=Q \Phi(a(\sigma))-N\Gamma,\label{MNGAMMAPhi}
\end{align}   
which is called the \textit{Generalized Sylvester Matrix Equation}. The explicit solutions can be found in \cite{zhou2005explicit}.
%\begin{rem}\myr The application of equation \eqref{XIQA-explicit} and the \textit{Generalized Sylvester Matrix Equation} \eqref{MNGAMMAPhi} significantly differs from that of Lemma 3.1 and Theorem 3.1 in \cite{xu2019generic}, as %both Lemma 3.1 and Theorem 3.1 of \cite{xu2019generic} 
%the latter results require the matrix $\Xi(a)$ to have no zero eigenvalue and to be of even dimension.
%\end{rem}

%resulting some generic internal model variants
%Then, we define a dynamic compensator as follows:
%\begin{subequations}\begin{align}
%\dot{\eta}_{1}=&M_1\eta_1+N_1x_2,\\
%\dot{\eta}_{2}=&M_2\eta_2+N_2u,
%\end{align}\end{subequations}

As shown in \cite{kreisselmeier2003nonlinear,marconi2007output,marconi2008uniform} and \cite{wang2023nonparametric}, there exists a continuous nonlinear mapping $\chi(\cdot)$ such that  
\begin{align}\label{IM-00}
\bm{\eta}^{\star}(v(t),\sigma,w))
&=\int_{-\infty}^t e^{M(t-\tau)}N \bm{\hat{x}}_2 (v(\tau), \sigma, w) d\tau,\\
\bm{\hat{x}}_2(v(t), \sigma, w))
&=\chi(\bm{\eta}^{\star}(v(t),\sigma,w)), ~~ \bm{\eta}^{\star}\in\mathds{R}^{n_0},\notag
%\hat{\bm{z}}^{\star}(t)=&Q^{-1}(v(t))\equiv 
\end{align}
that satisfies the differential equations 
\begin{align}\label{IM-01}
\frac{d \bm{\eta}^{\star}(v(t),\sigma,w)}{dt}  &= M\bm{\eta}^{\star}(v(t),\sigma,w) + N\bm{\hat{x}}_2(v(t), \sigma, w), \notag\\
\bm{\hat{x}}_2(v(t), \sigma, w)) &= \chi(\bm{\eta}^{\star}(v(t),\sigma,w)) ,
\end{align}
namely, the steady-state generator of $\bm{\hat{x}}_2$ with sufficiently large dimension $n_0$ and some continuous mapping $\chi(\cdot)$. Under Assumptions \ref{ass0}, \ref{ass0i}, and \ref{ass5-explicit}, insertion of \textit{Generalized Sylvester Matrix Equation} \eqref{MNGAMMAPhi} into \eqref{IM-00} and rearranging gives that $\bm{\eta}^{\star}=\theta$ (see Lemma 3 in \cite{wang2023nonparametric}).
Then, system \eqref{IM-01} leads the internal model
%\begin{subequations}\label{explicit-mas}
\begin{align}
\dot{\eta}&=M\eta+N\hat{x}_2,\label{explicit-mas1}%\\
%\dot{\eta}_2&=M_2\eta_2+N_2u.\label{explicit-mas2}%\\
%\dot{k}_{i} &= \rho_{i}(e_{vi})e_{vi}^{2} \label{Nussam-mas2}\\
%\dot{\hat{a}}_i  &=- k_i \Theta ^T(\eta_i )\left[\Theta (\eta_i )\hat{a}_i +\textnormal{\col}(\eta_{i,n_i +1},\cdots,\eta_{i,2n_i})\right] \\
%u &= -\rho(e, x_2- \chi_{1}(\eta_1) + \chi_{2}(\eta_2) \label{explicit-mas3}
\end{align}
%\end{subequations}
which is the internal model associated with the signal $\hat{x}_2$. 
\subsection{Error dynamics}

Perform coordinate and input transformations on the composite systems \eqref{eqn: exosystem system}, \eqref{nonlinear-systems-trans}, and \eqref{explicit-mas1} to give
\begin{align*}
\bar{z}&= z-\bm{z},&  \bar{x}= \tilde{x}-\bm{E}, \\
\bar{\eta}&= \eta-\bm{\eta}^{\star}-Nb^{-1}e,& e=y-\bm{x}_1, 
\end{align*}
which yields an error system in the form:
\begin{subequations}\label{Main-sys1}\begin{align}
\dot{\bar{z}} %=& f(z,y,v,w)-f(\bm{z},h(v,w),v,w)\nonumber\\
%=&f(\bar{z}+\bm{z},e+h(v,w),v,w)-f(\bm{z},h(v,w),v,w)\nonumber\\
&= \bar{f}(\bar{z},e,\mu),\\
\dot{\bar{x}}%= &\dot{\bm{E}}-\dot{\tilde{x}}\nonumber\\
%= &b^{-1}\dot{\bm{x}}-\dot{\bm{\hat{x}}}-\dot{\tilde{x}}\nonumber\\
%= &b^{-1}A_c\bm{x}+ b^{-1}g(\bm{z},\bm{x}_1,\mu)+B_c\bm{u}\nonumber\\
%&- A\bm{\hat{x}}-B_c\bm{u}- A\tilde{x}-b^{-1}(\lambda y+g(z,y,v,w))\nonumber\\
%= &b^{-1}A\bm{x}+b^{-1}\lambda \bm{x}_1+ b^{-1}g(\bm{z},\bm{x}_1,\mu)\nonumber\\
%&- A\bm{\hat{x}}- A\tilde{x}-b^{-1}(\lambda y+g(z,y,v,w))\nonumber\\
%= &A\big[b^{-1}\bm{x}-\bm{\hat{x}}-\tilde{x}\big]+b^{-1}\lambda( \bm{x}_1-y)\nonumber\\
%&+ b^{-1}\big(g(\bm{z},\bm{x}_1,\mu)-g(z,y,v,w)\big)\nonumber\\
%= &A\big[\bm{E}-\tilde{x}\big]-b^{-1}\lambda e\nonumber\\
%&+ b^{-1}\big(g(\bm{z},\bm{x}_1,\mu)-g(\bar{z}+\bm{z},e+\bm{x}_1,\mu)\big)\nonumber\\
&=A\bar{x}+ b^{-1}\!\big[\bar{g}(\bar{z}, e, \mu)+\lambda e\big],\\
\dot{\bar{\eta}}%=&\dot{\eta}-\dot{\bm{\eta}}^{\star}-Nb^{-1}\dot{e}\nonumber\\
%=&M\eta+N\hat{x}_2-M\bm{\eta}^{\star} - N\bm{\hat{x}}_2\nonumber\\
%&-N\hat{x}_2-N\bar{x}_2 +N\bm{\hat{x}}_2-Nb^{-1}\bar{g}_1(\bar{z}, e, \mu)\nonumber\\
%=&M\eta-M\bm{\eta}^{\star} -N\bar{x}_2 -Nb^{-1}\bar{g}_1(\bar{z}, e, \mu) \nonumber\\
&=M\bar{\eta} -N\!\left(\bar{x}_2 -b^{-1}e+b^{-1}\bar{g}_1(\bar{z}, e, \mu)\right)\!,\\
\dot{e}%=&b\hat{x}_2+b\tilde{x}_2+g_1(z,y,v,w)- \dot{\bm{x}}_1\nonumber \\
%=&b\hat{x}_2+b\tilde{x}_2+g_1(z,y,v,w)\nonumber\\
%&-   C_cA_c\bm{x}(\mu)-C_cg(\bm{z}(\mu),\bm{x}_1(\mu),\mu)-C_cB_cb\bm{u}(\mu)\nonumber\\
%=&b\hat{x}_2+b\tilde{x}_2-   \bm{x}_2(\mu)\nonumber\\
%&+g_1(z,y,v,w)-g_1(\bm{z}(\mu),\bm{x}_1(\mu),\mu)\nonumber\\
%=&b\hat{x}_2+b\tilde{x}_2-  b\bm{E}_2 -b\bm{\hat{x}}_2\nonumber\\
%&+g_1(\bar{z}+\bm{z},e+\bm{x}_1,v,w)-g_1(\bm{z},\bm{x}_1,\mu)\nonumber\\
%=&b\hat{x}_2+b\bar{x}_2 -b\bm{\hat{x}}_2+\bar{g}_1(\bar{z}, e, \mu) \nonumber\\
&=b(\hat{x}_2-\bm{\hat{x}}_2)+b\bar{x}_2 +\bar{g}_1(\bar{z}, e, \mu) ,\\
 \dot{\hat{x}}_i
 &= \hat{x}_{i+1}-\lambda_i\hat{x}_1,\quad i=2,\dots,r-1,\\
  \dot{\hat{x}}_r
  &= u-\lambda_r\hat{x}_1,
\end{align}\end{subequations}
where $\mu=\col(\sigma,v,w)$,
\begin{align*}
\bar{f}(\bar{z},e,\mu)
&=f(\bar{z}+\bm{z},e+\bm{x}_1,\mu)-f(\bm{z},\bm{x}_1,\mu),\\
\bar{g}(\bar{z}, e, \mu)
&= g(\bar{z}+\bm{z},e+\bm{x}_1,\mu)-g(\bm{z},\bm{x}_1,\mu).
\end{align*}
It can be verified that, for all $\mu\in \mathds{V}\times\mathds{W} \times\mathds{S}$, $\bar{f}(0,0,\mu)=0$ and $\bar{g}(0, 0, \mu)=0$. 
Problem \ref{Prob: second-order-Output-regulation} can be solved if a control law can be found to stabilize the system \eqref{Main-sys1}.

Let $\bar{x}_c=\col(\bar{x}, \bar{\eta})$ and $\bar{G}_c(\bar{z}, e, \mu)=b^{-1}\col\big(Ne-N\bar{g}_1(\bar{z}, e, \mu), \bar{g}(\bar{z}, e, \mu)+\lambda e\big)$; the system \eqref{Main-sys1} can be rewritten into the form
\begin{subequations}\label{Main-sys-com}\begin{align}
\dot{\bar{z}} 
&=\bar{f}(\bar{z},e,\mu),\\
\dot{\bar{x}}_c
&= \underbrace{\left[\begin{matrix}M & -NC_cA_c\\
\textbf{0}&A
\end{matrix}\right]}_{M_c}\bar{x}_c+\bar{G}_c(\bar{z}, e, \mu) ,\label{Main-sys-com-xc}\\
\dot{e}
&=b(\hat{x}_2-\chi(\bm{\eta}^*))+b\bar{x}_2 +\bar{g}_1(\bar{z}, e, \mu) ,\\
 \dot{\hat{x}}_i
 &= \hat{x}_{i+1}-\lambda_i\hat{x}_1, \quad i=2,\dots,r-1,\\
\dot{\hat{x}}_r
&=u-\lambda_r \hat{x}_1.
\end{align}\end{subequations}
%where $M_c=\left[\begin{matrix}M & -NC_cA_c\\
%0&A
%\end{matrix}\right]$. 
It can be verified that, for all $\mu\in \mathds{V}\times\mathds{W} \times\mathds{S}$,  $\bar{G}_c(0, 0, \mu)=\bm{0}$ and the matrix $M_c$ is Hurwitz. Hence, the $(\bar{z}, \bar{x}_c)$-subsystem in system \eqref{Main-sys-com} is in a similar form as the system (8) of \cite{wang2023nonparametric}. 
As a result,  the $(\bar{z}, \bar{x}_c)$-subsystem in system \eqref{Main-sys-com}, under Assumptions \ref{ass0}, \ref{ass1}, and \ref{H2}, admits the following properties (see Properties 1 and 2 in \cite{wang2023nonparametric}): 
%\begin{proposition}\label{Prop-1}%
\begin{prop}\label{property1}%
There exists a smooth input-to-state Lyapunov function $V_0\equiv V_0(\bar{z},\bar{x}_c)$ satisfying
\begin{align}
\underline{\alpha}_0(\|\bar{Z}\|)&\leq V_0(\bar{Z})\leq \bar{\alpha}_0(\|\bar{Z}\|),\notag\\
\dot{V}_0 &\leq  -\|\bar{Z}\|^2+\bar{\gamma}^*\bar{\gamma} \left(e\right),\label{V0}
\end{align}
for some positive constant $\bar{\gamma}^*$ and comparison functions $\underline{\alpha}_0(\cdot)\in \mathcal{K}_{\infty}$, $\bar{\alpha}_0(\cdot)\in \mathcal{K}_{\infty}$, and $\bar{\gamma}(\cdot)\in \mathcal{K}_\infty$ with $\bar{Z}=\textnormal{\col}(\bar{z},\bar{x}_c)$.
\end{prop}

\begin{prop}\label{property2}%
There are positive smooth functions $\gamma_{g0}(\cdot)$ and $\gamma_{g1}(\cdot)$ such that
$$b^2\bar{x}_2^2+\|\bar{g}_1(\bar{z}, e, \mu)\|^2\leq \gamma_{g0}(\bar{Z})\|\bar{Z}\|^2+e^2\gamma_{g1}(e).$$
\end{prop}
%\end{proposition}
\begin{rem}Since $\bar{G}_c(\bar{z}, e, \mu)$ in \eqref{Main-sys-com-xc} is smooth and satisfies $\bar{G}_c\!\left(0 ,0 ,\mu\right)=\bm{0}$, for all $\mu \in \mathds{V}\times\mathds{W} \times\mathds{S}$, by Lemma 7.8 in \cite{huang2004nonlinear},
$$\|P_c\bar{G}_c\!\big(\bar{z}, e ,v\big)\|^2\leq\pi_1 (\bar{z})\|\bar{z} \|^2+\phi_1(e )e^2$$
for some known smooth functions $\pi_1(\cdot)\geq 1$ and $\phi_1(\cdot)\geq 1$, where $P_c$ is a positive definite matrix such that $P_cM_c+M_cP_c^{\top}=-2I$. By Remark \ref{rem-barV}, for any smooth function $\Delta_{\bar{z}}(\bar{z})>0$, there exits a continuous function $\bar{V}_{\bar{z}}(\bar{z})$ satisfying 
$ \underline{\alpha}_{\bar{z}}(\|\bar{z}\|)\leq  \bar{V}_{\bar{z}}(\bar{z})\leq \bar{\alpha}_{\bar{z}}(\|\bar{z}\|)$ for some class $\mathcal{K}_{\infty}$ functions $\underline{\alpha}_{\bar{z}}(\cdot)$ and $\bar{\alpha}_{\bar{z}}(\cdot)$ such that, for any $\mu \in \mathds{V}\times\mathds{W} \times\mathds{S}$, the time derivative of $ \bar{V}_{\bar{z}}(\bar{z})$ along the trajectory \eqref{Reg-2} satisfies 
$$\dot{\bar{V}}_{\bar{z}}(\bar{z})\leq -\Delta_{\bar{z}}(\bar{z})\|\bar{z}\|^2+\delta_{\bar{z}}\gamma_{\bar{z}}(e)e^2,$$
where $\delta_{\bar{z}}$ and $\gamma_{\bar{z}}(\cdot)$ are some positive constant and positive function. Let $V_0(\bar{Z})=\bar{V}_{\bar{z}}(\bar{z})+\bar{x}_c^{\top}P_c\bar{x}_c$, which satisfies  $\underline{\alpha}_0(\|\bar{Z}\|)\leq V_0(\bar{Z})\leq \bar{\alpha}_0(\|\bar{Z}\|)$ for some class $\mathcal{K}_{\infty} $ functions $\underline{\alpha}_0(\cdot)\in \mathcal{K}_{\infty}$ and $\bar{\alpha}_0(\cdot)\in \mathcal{K}_{\infty}$. By choosing $\Delta_{\bar{z}}(\bar{z})>\pi_1 (\bar{z})+1$, the time derivative of $V_0(\bar{Z})$ along the $\bar{Z}$-subsystem of \eqref{Main-sys-com} satisfies 
\begin{align*}
    \dot{V}_0 \leq& -\Delta_{\bar{z}}(\bar{z})\|\bar{z}\|^2+\delta_{\bar{z}}\gamma_{\bar{z}}(e)e^2-\|x_c\|^2 +\|P_c\bar{G}_c\!\big(\bar{z}, e ,v\big)\|^2\\
   % \leq&\ -\Delta_{\bar{z}}(\bar{z})\|\bar{z}\|^2+\delta_{\bar{z}}\gamma_{\bar{z}}(e)e^2-\|x_c\|^2\\
   % &\ +\pi_1 (\bar{z})\|\bar{z} \|^2+\phi_1(e )e^2\\
   % \leq&\ -(\Delta_{\bar{z}}(\bar{z})-\pi_1 (\bar{z}))\|\bar{z}\|^2-\|x_c\|^2\\
   % &\ +\underbrace{(\delta_{\bar{z}}\gamma_{\bar{z}}(e)+\phi_1(e ))e^2}_{\bar{\gamma}^*\bar{\gamma} \left(e\right)}\\
    \leq& -(\underbrace{\Delta_{\bar{z}}(\bar{z})-\pi_1 (\bar{z})}_{> 1})\|\bar{z}\|^2-\|x_c\|^2 +\underbrace{(\delta_{\bar{z}}\gamma_{\bar{z}}(e)+\phi_1(e ))e^2}_{\bar{\gamma}^*\bar{\gamma} \left(e\right)}.
\end{align*}
Since $b^2\bar{x}_2^2+\|\bar{g}_1(\bar{z}, e, \mu)\|^2$ is smooth and vanishes at $\bm{0}$ when $\textnormal{\col}(\bar{x}_2,\bar{z}, e)=\textnormal{\col}\left(0 ,0 ,0\right)$, for all $\mu \in \mathds{V}\times\mathds{W} \times\mathds{S}$, by using Lemma 7.8 in \cite{huang2004nonlinear}, Property \ref{property2} can be verified.
\end{rem}

\section{Main results}\label{mainresults}
\subsection{Non-adaptive method in robust output regulation}

\begin{figure}[htp]
    \centering
\tikzstyle{block} = [draw, fill=mitSilverGray!10, rectangle,
    minimum height=1.5em, minimum width=6em]
    \tikzstyle{controllaw} = [draw, fill=mitBlue!60, rectangle,
    minimum height=1.5em, minimum width=1em]
      \tikzstyle{controllaw2} = [draw, fill=mitBlue!60, rectangle,
    minimum height=1.5em, minimum width=2em]
    \tikzstyle{iterative} = [draw, fill=mitRed!60, rectangle,
    minimum height=1.5em, minimum width=2em]
    \tikzstyle{iterative2} = [draw=black, fill=mitGreen!60, rectangle,
    minimum height=1.5em, minimum width=2em]
    \tikzstyle{interModel} = [draw, fill=mitPurple!60, rectangle,
    minimum height=1.5em, minimum width=6em]
\tikzstyle{sum} = [draw, fill=green!20, circle, node distance=1cm,minimum size=0.05cm]
\tikzstyle{input} = [coordinate]
\tikzstyle{output} = [coordinate]
\tikzstyle{pinstyle} = [pin edge={to-,thin,black}]
% The block diagram code is probably more verbose than necessary
\begin{tikzpicture}[scale=0.775,every node/.style={transform shape}, node distance=2cm,>=latex']
    % We start by placing the blocks
    \node [input, name=input] {};
    \node [sum, fill=mitLightBlue, right =1.625cm of input] (sum) {};
    \node [controllaw, right = 1cm of sum] (controller) {~$u = \alpha_r(\epsilon_1,\epsilon_2,\dots,\epsilon_r,k^*,\eta,\hat{x}_1)$~};
    \node [controllaw2, below left = 3cm and 0.025cm of sum] (controller-ttle) {Controller};
    \node [interModel, right = 0.15cm of controller-ttle] (interModel-ttle) {Internal Model};
    \node [iterative, right  = 0.15cm of interModel-ttle] (iterative-ttle) {Input-driven Filter};
    \node [iterative2, right = 0.15cm of iterative-ttle] (recursive-ttle) {Recursive Equations};
    \node [block, right = 1cm of controller,
            node distance=3cm, pin={[pinstyle]below:Disturbances}, fill=mitSilverGray!60] (system) {System};
    % We draw an edge between the controller and system block to
    % calculate the coordinate u. We need it to place the measurement block.
    \draw [thick,->] (controller) -- node[name=u,above] {$u$} (system);
  %  \node [output, right =0.3cm of system] (output) {};
   %  \node [output, right=0.6cm of output] (output1) {};
    \node [output,above =0.5cm of controller] (measurements) {};
    \node [output, right =0.1cm of measurements] (measurements2) {};
   
     \node [iterative, below left =0.85cm and -4.95cm of controller] (Ivlearn) { $\dot{\hat{x}}=A\hat{x}+B_cu$};
     \node [output,left =1.8825cm of Ivlearn] (Ivlearn1) {};

      \node [interModel, below =0.5cm of Ivlearn] (IMmodel) {$\dot{\eta}=M\eta+N\hat{x}_2$};
    %\node [iterativeframe, below right =0.265cm and -3.6000 of IMmodel] (iterativeframe) { };
    \node [output,left =0.35cm of IMmodel] (IMmodel2) {};
    \node [output,left =1.455cm of IMmodel2] (IMmodel3) {};
    
      \node [iterative2, below=0.65cm of sum] (ErrorTran) {$\begin{matrix}\alpha_1( \epsilon_1,k^*,\eta)\\
      \alpha_2(\epsilon_1,\epsilon_2,k^*,\eta,\hat{x}_1)\\
      \vdots\\
      \alpha_{r}(\epsilon_1,\dots,\epsilon_r,k^*,\eta,\hat{x}_1)
      \end{matrix}$};
      \node [right=1.8cm of ErrorTran.north east] (ErrorTran1){};
    % Once the nodes are placed, connecting them is easy.
    \draw [draw,thick,->] (input) -- node[above] { $h(v,w)$} (sum);
   % \draw [->] (sum) --  (controller);
    %\draw [->] (system) -- (output1);
    \draw [thick,->] (system.north)|-(measurements2)--(measurements)node[below] {$y$}-| node[pos=0.95,right] {$-$}(sum);
    \draw [thick,->] (u)|-(Ivlearn);
    \draw [thick,->] (IMmodel)--node[pos=0.5, below] {$\eta$}(IMmodel3);%-|(5.225-0.25,-0.3);%
    \draw [thick,->] (Ivlearn)--node[pos=0.5, right] {$\hat{x}_2$}(IMmodel);
    \draw [thick,->] (Ivlearn)--node[pos=0.5, below] {$\hat{x}$}(Ivlearn1);
    \draw [thick,->] (sum)--node[pos=0.5,left]{$\epsilon_1=e$}(ErrorTran);
    \draw [thick,->] (ErrorTran.north east)--node[pos=0.5,above]{$\alpha_r$}(ErrorTran1)-|(controller.south);
\end{tikzpicture}
\caption{Non-adaptive method in robust output regulation.}\label{framework}
\end{figure}
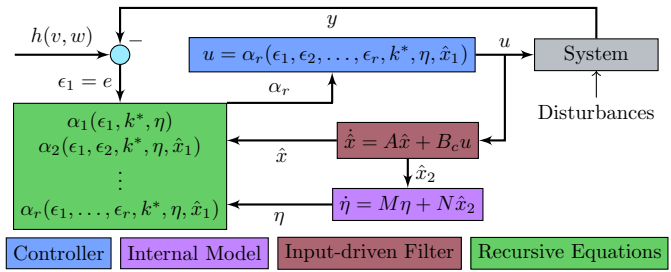

The proposed non-adaptive framework for the solution of the robust output regulation problem is shown in Fig.~\ref{framework}. The analysis presented below is based on the backstepping recursive method introduced in \cite{krstic1995nonlinear}. This technique was further generalized to neural network control for strict-feedback nonlinear systems in \cite{zhang2000adaptive}. This approach is suitable for the design of control systems that can handle the complexities and nonlinearities of the error system \eqref{Main-sys-com}. 
Backstepping is a widely used approach for analyzing lower triangular systems and remains an active area of research in the control community. 
Very recently, \cite{duan2021highIII,duan2021highIV} originally established the high-order fully actuated system approaches, and based on which the second- and higher-order methods of backstepping (recursive method) are effectively proposed for both uncertain second-order and higher-order strict-feedback nonlinear systems. 
%r example, \cite{duan2021highIII,duan2021highIV} effectively employed high-order fully actuated system approaches to generalize the second- and high-order methods of backstepping (recursive method) study both uncertain second-order and higher-order strict-feedback nonlinear systems.}
%
By iterating the control design process, the recursive method ensures the convergence, robustness, and stability of the error system \eqref{Main-sys-com}. The analysis adopts the  expressions: %$\epsilon_1=e$ and $\epsilon_{i+1}=\hat{x}_{i+1}-\alpha_i(\epsilon_1,\dots,\epsilon_i,k^*,\eta,\hat{x}_1)$ with 
\begin{align}\label{alpha-function}
 \epsilon_1= e,~~\epsilon_{i+1}=&\ \hat{x}_{i+1}-\alpha_i(\epsilon_1,\dots,\epsilon_i,k^*,\eta,\hat{x}_1),\nonumber\\
\alpha_1( \epsilon_1,k^*,\eta)=&-k^*\rho(\epsilon_1)\epsilon_1+\chi(\eta),\nonumber\\
\alpha_2(\epsilon_1,\epsilon_2,k^*,\eta,\hat{x}_1)%=&-b\epsilon_1-\epsilon_2+\lambda_{2}\hat{x}_1+\frac{\partial \alpha_1}{\partial \epsilon_1}\dot{e}_1+\frac{\partial \alpha_1 }{\partial \eta}\dot{\eta}\nonumber\\
%&+\frac{\partial \alpha_1}{\partial e}\Big[b(\hat{x}_2-\bm{\hat{x}}_2)+b\bar{x}_2 +\bar{g}_1(\bar{z}, e, \mu)\Big]\nonumber\\
=&-b\epsilon_1-\epsilon_2+\lambda_{2}\hat{x}_1+\frac{\partial \alpha_1 }{\partial \eta}\dot{\eta}\nonumber\\
&+b\frac{\partial \alpha_1}{\partial \epsilon_1}(\epsilon_2-k^*\rho(\epsilon_1)\epsilon_1)\nonumber\\
&-\frac{1}{2}\epsilon_2\!\left(\frac{\partial \alpha_1}{\partial \epsilon_1}\right)^{\!\!2}+\frac{\partial \alpha_1 }{\partial k^*}\dot{k}^*,\nonumber\\
%\alpha_3(\epsilon_1,\epsilon_2,\epsilon_{3},k^*,\eta,\hat{x}_1)=&-e_{2}-\epsilon_{3}+\lambda_{3}\hat{x}_{1}+\frac{\partial \alpha_{2} }{\partial \eta}\dot{\eta}\nonumber\\
%&+\frac{\partial \alpha_2 }{\partial \epsilon_2}\dot{e}_2+\frac{\partial \alpha_2 }{\partial \hat{x}_1}\dot{\hat{x}}_1\nonumber\\
%&+b\frac{\partial \alpha_2}{\partial \epsilon_1}(\hat{x}_2-\chi(\eta))-\frac{1}{2}\epsilon_{3}\left(\frac{\partial \alpha_2}{\partial \epsilon_1}\right)^2\nonumber\\
\alpha_i(\epsilon_1,\dots,\epsilon_i,k^*,\eta,\hat{x}_1)=&-\epsilon_{i-1}-\epsilon_i+\lambda_{i}\hat{x}_{1}+\frac{\partial \alpha_{i-1} }{\partial \eta}\dot{\eta}\nonumber\\
&+\frac{\partial \alpha_{i-1} }{\partial \hat{x}_1}\dot{\hat{x}}_1+ \sum_{j=2}^{i-1}\frac{\partial \alpha_{i-1} }{\partial \epsilon_{j}}\dot{\epsilon}_j\nonumber\\
 &+b\frac{\partial \alpha_{i-1}}{\partial \epsilon_1}(\epsilon_2-k^*\rho(\epsilon_1)\epsilon_1)\nonumber\\
 &-\frac{1}{2}\epsilon_i\!\left(\frac{\partial \alpha_{i-1}}{\partial \epsilon_1}\right)^{\!\!2}+\frac{\partial \alpha_{i-1} }{\partial k^*}\dot{k}^*,\;\nonumber\\
 & \qquad\qquad\qquad i=3,\dots,r, 
\end{align}
where $\dot{k}^*$ will be zero when $k^*$ is a constant, and $\hat{x}_{1}, \dots, \hat{x}_{r}$ and $\eta$ are generated in \eqref{input-filter} and \eqref{explicit-mas1}, respectively.
\begin{thm}\label{Theorem-1}%
For the system \eqref{Main-sys-com} under Assumptions \ref{ass0}--\ref{ass5-explicit}, there is a sufficiently large positive smooth function $\rho(\cdot)$ and a positive real number $k^*$ such that the controller 
\begin{align}
u &= \alpha_r(\epsilon_1,\epsilon_2,\dots,\epsilon_r,k^*,\eta,\hat{x}_1)\label{ESC-1b}
\end{align}
 solves Problem \ref{Prob: second-order-Output-regulation}. In addition, there exists a continuous
positive definite function $U_r(\bar{Z}, \epsilon_1,\dots, \epsilon_r)$
such that, for all $\mu\in \mathds{S}\times \mathds{V}\times \mathds{W}$,
\begin{align}\label{dotV}\dot{U}_r(\bar{Z}, \epsilon_1,\dots, \epsilon_r)\leq -\big\|\bar{Z}\big\|^{2}-\sum\nolimits_{j=1}^{r}\epsilon_{j}.\end{align}
%for any $k_0\geq k^*$. 
\end{thm}%

\begin{proof} From Property \ref{property1}, the changing supply rate technique \cite{sontag1995changing} can be applied to show that, given any smooth function $\Delta_{Z}(\bar{Z})>0$, there exists a continuous function $V_{1}(\bar{Z})$ satisfying
$$\underline{\alpha}_{1}\big(\big\|\bar{Z} \big\|^2\big)\leq V_{1}\big( \bar{Z}  \big)\leq\overline{\alpha}_{1}\big(\big\|\bar{Z} \big\|^2\big)$$
 for some class $\mathcal{K}_{\infty}$ functions $\underline{\alpha}_{1}(\cdot)$ and $\overline{\alpha}_{1}(\cdot)$, such that, for all $\mu\in \Sigma$, along the trajectories of the $Z $ subsystem, $$\dot{V}_{1} \leq-\Delta_{Z}(\bar{Z} )\big\|\bar{Z} \big\|^2+ \hat{\gamma}^* \hat{\gamma} \left(\epsilon_1\right)\epsilon_1^2, $$
where $\hat{\gamma}^*$ is known positive constant and $\hat{\gamma} \left(\cdot\right)\geq 1$ is a known smooth positive definite function.

Define the Lyapunov function $U_1(\bar{Z}, \epsilon_1)=V_{1}\big( \bar{Z}  \big)+ \epsilon_1^2$. Then, the time derivative of $U_1\equiv U_1(\bar{Z}, \epsilon_1)$ along the trajectory of $\epsilon_1$-subsystem with $\hat{x}_2=\epsilon_2+\alpha_1$ and $\eta=\bar{\eta}+\bm{\eta}^{\star}+Nb^{-1}\epsilon_1 $ leads to
%\begin{align}\label{U1-derivative}
%\dot{U}_1(\bar{Z}, e)%=& \dot{V}_{1}\big( \bar{Z}  \big)+ 2\epsilon_1\dot{e}_1\nonumber\\
%\leq &-\Delta_{Z}(\bar{Z} )\big\|\bar{Z} \big\|^2+ \hat{\gamma}^* \hat{\gamma} \left(\epsilon_1\right)\epsilon_1^2\nonumber\\
%&+2b\epsilon_1(\hat{x}_2-\chi(\bm{\eta}^*))+2b\epsilon_1\bar{x}_2 +2\epsilon_1\bar{g}_1(\bar{z}, \epsilon_1, \mu)\nonumber\\
%\leq &-\Delta_{Z}(\bar{Z} )\big\|\bar{Z} \big\|^2-\big(2b k^*\rho(\epsilon_1)-\hat{\gamma}^* \hat{\gamma} \left(\epsilon_1\right)\big)\epsilon_1^2\nonumber\\
%&+2b\epsilon_1(\epsilon_2-\bar{\chi}(\bar{\eta}, \epsilon_1, \mu))\\
%&+2b\epsilon_1\bar{x}_2 +2\epsilon_1\bar{g}_1(\bar{z}, \epsilon_1, \mu)\nonumber\\
%\leq &-\Delta_{Z}(\bar{Z} )\big\|\bar{Z} \big\|^2-\big(2b k^*\rho(\epsilon_1)-\hat{\gamma}^* \hat{\gamma} \left(\epsilon_1\right)\big)\epsilon_1^2\nonumber\\
%&+2b\epsilon_1\epsilon_2-2b\epsilon_1\bar{\chi}(\bar{\eta}, \epsilon_1, \mu)\\
%&+2b\epsilon_1\bar{x}_2 +2\epsilon_1\bar{g}_1(\bar{z}, \epsilon_1, \mu)\nonumber\\
%\leq &-\Delta_{Z}(\bar{Z} )\big\|\bar{Z} \big\|^2\nonumber\\
%&-\big(2b k^*\rho(\epsilon_1)-3-\hat{\gamma}^* \hat{\gamma} \left(\epsilon_1\right)\big)\epsilon_1^2\nonumber\\
%&+2b\epsilon_1\epsilon_2+\Delta_1( \epsilon_1,\bar{Z},\mu)
%\end{align}
\begin{align}\label{U1-derivative}
\dot{U}_1(\bar{Z}, \epsilon_1)=&\ \dot{V}_{1}\big( \bar{Z}  \big)+ 2\epsilon_1\dot{\epsilon}_1\nonumber\\
\leq &-\Delta_{Z}(\bar{Z} )\big\|\bar{Z} \big\|^2+ \hat{\gamma}^* \hat{\gamma} \left(\epsilon_1\right)\epsilon_1^2+2\epsilon_1\bar{g}_1(\bar{z}, \epsilon_1, \mu) \nonumber\\
&+2b\epsilon_1(\underbrace{\epsilon_2+\alpha_1( \epsilon_1,k^*,\eta)}_{\hat{x}_2}-\chi(\bm{\eta}^*))+2b\epsilon_1\bar{x}_2\nonumber\\
\leq &-\Delta_{Z}(\bar{Z} )\big\|\bar{Z} \big\|^2+ \hat{\gamma}^* \hat{\gamma} \left(\epsilon_1\right)\epsilon_1^2+2\epsilon_1\bar{g}_1(\bar{z}, \epsilon_1, \mu)\nonumber\\
&+2b\epsilon_1(\epsilon_2+\underbrace{\alpha_1( \epsilon_1,k^*,\eta)}_{-k^*\rho(\epsilon_1)\epsilon_1+\chi(\eta)}-\chi(\bm{\eta}^*))+2b\epsilon_1\bar{x}_2 \nonumber\\
\leq &-\Delta_{Z}(\bar{Z} )\big\|\bar{Z} \big\|^2-\big(2b k^*\rho(\epsilon_1)-\hat{\gamma}^* \hat{\gamma} \left(\epsilon_1\right)\big)\epsilon_1^2\nonumber\\
&+2b\epsilon_1(\epsilon_2-\bar{\chi}(\bar{\eta},\epsilon_1,\mu))\nonumber\\
&+2b\epsilon_1\bar{x}_2 +2\epsilon_1\bar{g}_1(\bar{z}, \epsilon_1, \mu)\nonumber\\
\leq &-\Delta_{Z}(\bar{Z} )\big\|\bar{Z} \big\|^2-\big(2b k^*\rho(\epsilon_1)-\hat{\gamma}^* \hat{\gamma} \left(\epsilon_1\right)\big)\epsilon_1^2\nonumber\\
&+2b\epsilon_1\epsilon_2-2b\epsilon_1\bar{\chi}(\bar{\eta},\epsilon_1,\mu)\nonumber\\
&+2b\epsilon_1\bar{x}_2 +2\epsilon_1\bar{g}_1(\bar{z}, \epsilon_1, \mu)\nonumber\\
\leq &-\Delta_{Z}(\bar{Z} )\big\|\bar{Z} \big\|^2-\big(2b k^*\rho(\epsilon_1)-3-\hat{\gamma}^* \hat{\gamma} \left(\epsilon_1\right)\big)\epsilon_1^2\nonumber\\
&+2b\epsilon_1\epsilon_2+\Delta_1( \epsilon_1,\bar{Z},\mu)
\end{align}
where 
\begin{align*}
    \Delta_1( \epsilon_1,\bar{Z},\mu)&=b^2\bar{x}_2^2 +\bar{g}_1(\bar{z}, \epsilon_1, \mu)^2+b^2\bar{\chi}(\bar{\eta}, \epsilon_1, \mu)^2,\\
    \bar{\chi}(\bar{\eta}, \epsilon_1, \mu) &\equiv \chi(\bar{\eta}+\bm{\eta}^*+Nb^{-1} \epsilon_1)-\chi(\bm{\eta}^*).
\end{align*}
Now let $ U_2(\bar{Z}, \epsilon_1, \epsilon_2)=U_1(\bar{Z}, \epsilon_1)+\epsilon_2^2$. The time derivative of $U_2\equiv U_2(\bar{Z}, \epsilon_1, \epsilon_2)$ along the trajectory of $\epsilon_2$-subsystem with $\hat{x}_3=\epsilon_{3}+\alpha_2$ is given by 
\begin{align}%\label{U1-derivative-1}
\dot{U}_2 \leq &\ \dot{U}_1+2 \epsilon_2\dot{\epsilon}_2\nonumber\\
\leq&-\Delta_{Z}(\bar{Z} )\big\|\bar{Z} \big\|^2-\big(b k^*\rho(\epsilon_1)-3-\hat{\gamma}^* \hat{\gamma} \left(\epsilon_1\right)\big)\epsilon_1^2\nonumber\\
&+2b\epsilon_1\epsilon_2+\Delta_1( \epsilon_1,\bar{Z},\eta)+2 \epsilon_2(\epsilon_{3}+\alpha_2-\lambda_2\hat{x}_1-\dot{\alpha}_1)\nonumber\\
\leq&-\Delta_{Z}(\bar{Z} )\big\|\bar{Z} \big\|^2-\big(b k^*\rho(\epsilon_1)-3-\hat{\gamma}^* \hat{\gamma} \left(\epsilon_1\right)\big)\epsilon_1^2\nonumber\\
&+2b\epsilon_1\epsilon_2+\Delta_1( \epsilon_1,\bar{Z},\eta)\nonumber\\
&+2 \epsilon_2\Big(\epsilon_{3}+\alpha_2-\lambda_2\hat{x}_1-\frac{\partial \alpha_1}{\partial \epsilon_1}\dot{\epsilon}_1-\frac{\partial \alpha_1 }{\partial \eta}\dot{\eta}-\frac{\partial \alpha_1 }{\partial k^*}\dot{k}^*\Big)\nonumber\\
\leq&-\Delta_{Z}(\bar{Z} )\big\|\bar{Z} \big\|^2-\big(b k^*\rho(\epsilon_1)-3-\hat{\gamma}^* \hat{\gamma} \left(\epsilon_1\right)\big)\epsilon_1^2\nonumber\\
&+2b\epsilon_1\epsilon_2+\Delta_1( \epsilon_1,\bar{Z},\eta)+2 \epsilon_2\epsilon_{3}\nonumber\\
&+2 \epsilon_2\Big(\alpha_2-\lambda_2\hat{x}_1-\frac{\partial \alpha_1 }{\partial \eta}\dot{\eta}-b\frac{\partial \alpha_1}{\partial \epsilon_1}(\underbrace{\epsilon_2+\alpha_1}_{\hat{x}_2}-\chi(\eta))\nonumber\\
&-\frac{\partial \alpha_1}{\partial \epsilon_1}\big[b\bar{\chi}(\bar{\eta},\epsilon_1,\mu)+b\bar{x}_2 +\bar{g}_1(\bar{z}, \epsilon_1, \mu)\big]-\frac{\partial \alpha_1 }{\partial k^*}\dot{k}^*\Big)\nonumber\\
\leq&-\Delta_{Z}(\bar{Z} )\big\|\bar{Z} \big\|^2-\big(b k^*\rho(\epsilon_1)-3-\hat{\gamma}^* \hat{\gamma} \left(\epsilon_1\right)\big)\epsilon_1^2\nonumber\\
&+\Delta_1( \epsilon_1,\bar{Z},\eta)+2 \epsilon_2\epsilon_{3}\nonumber\\
&+2 \epsilon_2\Big(\alpha_2-\epsilon_2+ \underbracea{\epsilon_2+b\epsilon_1-\lambda_2\hat{x}_1-\frac{\partial \alpha_1 }{\partial \eta}\dot{\eta}}\nonumber\\
&\underbraceb{{}-b\frac{\partial \alpha_1}{\partial \epsilon_1}(\epsilon_2\underbrace{-k^*\rho(\epsilon_1)\epsilon_1+\chi(\eta)}_{\alpha_1}-\chi(\eta))-\frac{\partial \alpha_1 }{\partial k^*}\dot{k}^*{}}_{-\alpha_2 }\nonumber\\
&\underbraced{{}+\frac{1}{2}\epsilon_2\Big(\frac{\partial \alpha_1}{\partial \epsilon_1}\Big)^2}\Big)+\epsilon^2_2\nonumber\\
&+\underbrace{\big[b^2\bar{\chi}(\bar{\eta},\epsilon_1,\mu)^2+b^2\bar{x}_2^2 +\bar{g}_1(\bar{z}, \epsilon_1, \mu)^2\big]}_{\Delta_1( \epsilon_1,\bar{Z},\eta)}\nonumber\\
\leq&-\Delta_{Z}(\bar{Z} )\big\|\bar{Z} \big\|^2+2\Delta_1( \epsilon_1,\bar{Z},\mu)+2 \epsilon_2\epsilon_{3}\nonumber\\
&-\big(2b k^*\rho(\epsilon_1)-3-\hat{\gamma}^* \hat{\gamma} \left(\epsilon_1\right)\big)\epsilon_1^2-\epsilon_2^2.\nonumber
\end{align}

Now let $U_i(\bar{Z}, \epsilon_1, \dots,\epsilon_i)=U_{i-1}(\bar{Z}, \epsilon_1,\dots,\epsilon_{i-1})+\epsilon_i^2$. The time derivative of $U_i\equiv U_i(\bar{Z}, \epsilon_1,\dots,e_{i})$ along the trajectory of $\epsilon_i$-subsystem with $\hat{x}_{i+1}=\epsilon_{i+1}+\alpha_i$ is given by 
\begin{align}%\label{U1-derivative-i}
\dot{U}_i\leq&-\Delta_{Z}(\bar{Z} )\big\|\bar{Z} \big\|^2+i\Delta_1( \epsilon_1,\bar{Z},\mu)+2 \epsilon_i\epsilon_{i+1}\nonumber\\
&-\big(2b k^*\rho(\epsilon_1)-3-\hat{\gamma}^* \hat{\gamma} \left(\epsilon_1\right)\big)\epsilon_1^2-\sum\nolimits_{j=2}^{i}\epsilon_{j}^2.\nonumber
\end{align}
Finally, at $i=r$ and ${\epsilon}_{r+1}=0$ results in
\begin{align}\label{U1-derivative-r}
\dot{U}_r\leq&-\Delta_{Z}(\bar{Z} )\big\|\bar{Z} \big\|^2+r\Delta_1( \epsilon_1,\bar{Z},\mu)\nonumber\\
&-\big(2b k^*\rho(\epsilon_1)-3-\hat{\gamma}^* \hat{\gamma} \left(\epsilon_1\right)\big)\epsilon_1^2-\sum\nolimits_{j=2}^{r}\epsilon_{j}^2.
\end{align}
From \cite{kreisselmeier2003nonlinear}, $\chi(\cdot)$ is a continuously differentiable function defined in \eqref{IM-00}. 
Moreover, it can be verified that the function $\bar{\chi}(\bar{\eta},\epsilon_1,\mu)$ is continuous and vanishes at $\col(\bar{z}, \epsilon_1, \bar{\eta})=\col(0,0,0)$ for all $\mu \in \mathds{V}\times\mathds{W} \times\mathds{S}$. As a result, the function $\Delta_1( \epsilon_1, \bar{Z}, \mu)=b^2\bar{x}_2^2 +\bar{g}_1(\bar{z}, \epsilon_1, \mu)^2+b^2\bar{\chi}(\bar{\eta},\epsilon_1,\mu)^2$ is continuous differentiable and vanishes at $\col(\bar{Z}, \epsilon_1, \mu)=\col(0,0,0)$ for all $\mu \in \mathds{V}\times\mathds{W} \times\mathds{S}$.
Following Lemma 11.1 of \cite{chen2015stabilization}, there exist positive smooth functions $\gamma_{1}(\cdot)$ and $\gamma_{2}(\cdot)$ such that
\begin{align*}\|\Delta_1( \epsilon_1,\bar{Z},\mu)\|^2\leq&  \gamma_{1}(\bar{Z})\|\bar{Z}\|^2 + \epsilon_1^2\gamma_{2}(\epsilon_1)
\end{align*}
for $\mu\in \mathds{V}\times \mathds{W}\times \mathds{S}$.  
We can then choose the functions $\Delta_{Z}(\bar{Z} )$ and $\rho(\epsilon_1)$ and the constant $k^*$ as
\begin{align}\label{rho-inquality-1}
\Delta_{Z}(\bar{Z} ) \geq & \ \gamma_{1}(\bar{Z})+1,\notag\\
\rho(\epsilon_1)\geq & \ \max\{\gamma_{2}(\epsilon_1), \hat{\gamma} \left(\epsilon_1\right), 1\},\\
 k^*\geq& \ {(3+\hat{\gamma}^*)}/{(2b)},\notag
\end{align}
such that \eqref{dotV} is satisfied. That is, for all $\mu \in \mathds{V}\times\mathds{W}\times\mathds{S}$, the equilibrium of the closed-loop system at the origin is
globally asymptotically stable. This completes the proof.
\end{proof}

From Theorem \ref{Theorem-1}, we can also use the adaptive method to estimate the $k^*$.
\begin{Corollary}\label{Theorem-2}%
For the system \eqref{Main-sys-com} under Assumptions \ref{ass0}--\ref{ass5-explicit}, there is a sufficiently large enough positive smooth function $\rho(\cdot)$ such that the controller, 
\begin{subequations}\label{adapha}\begin{align}
u &= \alpha_r(\epsilon_1,\epsilon_2,\dots,\epsilon_r,\hat{k},\eta), \label{adESC-1b}\\
\dot{\hat{k}}&=\rho(\epsilon_1)\epsilon_1^2,\label{adESC-1bb}
\end{align}\end{subequations}
 solves Problem \ref{Prob: second-order-Output-regulation} with the functions $\alpha_1( \epsilon_1,\hat{k},\eta)$, $\alpha_2(\epsilon_1,\epsilon_2,\hat{k},\eta,\hat{x}_1)$ and $\alpha_i(\epsilon_1,\dots,\epsilon_i,\hat{k},\eta,\hat{x}_1)$ defined in \eqref{alpha-function}, for $i=3,\dots,r$.
\end{Corollary}%
\begin{rem} The proof of Corollary \ref{Theorem-2} can easily proceed with the Lyapunov function 
$$V_r(\bar{Z},\epsilon_1,\dots,\epsilon_r, \hat{k}-k^*)=U_r(\bar{Z}, \epsilon_1,\dots, \epsilon_r)+b(\hat{k}-k^*)^2.$$
Therefore, the proof is omitted for the sake of brevity.
By differentiating $V_r(t)$
and Theorem~\ref{Theorem-1}, we can show that $\dot{V}_r(t)\le 0$.
Thus $V_r(t)$ is nonincreasing and bounded below, and therefore
converges as $t\to\infty$. 
In the meantime, the updated law in \eqref{adESC-1bb} does not generate an unbounded
high-gain. 
Since $\dot{V}_r(t)$ is uniformly continuous
under the smooth closed-loop dynamics, Barbalat’s Lemma yields
$\dot{V}_r(t)\to 0$, implying that the driving term of $\dot{\hat{k}}$
vanishes asymptotically. Consequently, $\hat{k}(t)$ is uniformly bounded
and converges to a finite limit.
Notably, the control law \eqref{adapha} differs from \cite{liu2009parameter} and \cite{tomei2023adaptive} by utilizing a non-adaptive design framework.
\end{rem}

\begin{rem} 
 Following Lemma~3 in \cite{wang2023nonparametric}, the existence of the nonlinear mapping
$\chi$ in \eqref{IM-01} hinges on solving the time-varying linear equation
\[
\Theta(\eta)\,\check{a}(\eta) + \operatorname{col}(\eta_{n+1},\dots,\eta_{2n}) = 0.
\]
Since $\Theta(\eta(t))$ may be singular at isolated time instants, we introduce a globally
defined smooth approximate inverse $O(\Theta):\mathbb{R}^{n\times n}\to\mathbb{R}^{n\times n}$ and set
\[
\check{a}(\eta) \triangleq -\,O(\Theta(\eta))\,\operatorname{col}(\eta_{n+1},\dots,\eta_{2n}),
\]
with
\begin{equation}\label{eq-Theta}
  O(\Theta)
  = \frac{\det(\Theta)}{\det(\Theta)^2 + \Psi\!\big(1+\det(\Theta)^2 - \varepsilon^2\big)}\,\operatorname{adj}(\Theta),
\end{equation}
where $\varepsilon>0$ is a design constant, $\operatorname{adj}(\cdot)$ denotes the adjugate matrix, and
\[
\Psi(\varsigma)=\frac{\kappa(\varsigma)}{\kappa(\varsigma)+\kappa(1-\varsigma)},\qquad
\kappa(s)=\begin{cases} e^{-1/s}, & s>0,\\[2pt] 0, & s\le 0.\end{cases}
\]
Note that $O(\Theta)=\Theta^{-1}$ when $\Theta$ is nonsingular and $\det(\Theta)^2\gg\varepsilon^2$,
whereas $O(\Theta)\to 0$ smoothly as $\det(\Theta)\to 0$, which ensures global smoothness.
Finally, define
\[
\chi(\eta)=\Gamma\,\Xi \big(\check{a}(\eta)\big)\,\operatorname{col}(\eta_{1},\dots,\eta_{n}).
\]
\end{rem}

\section{Application to Duffing's system}\label{numerexam}
%\subsection{Example 1: Application to Duffing's system}
\begin{figure}%[htbp]
\centering
\epsfig{figure=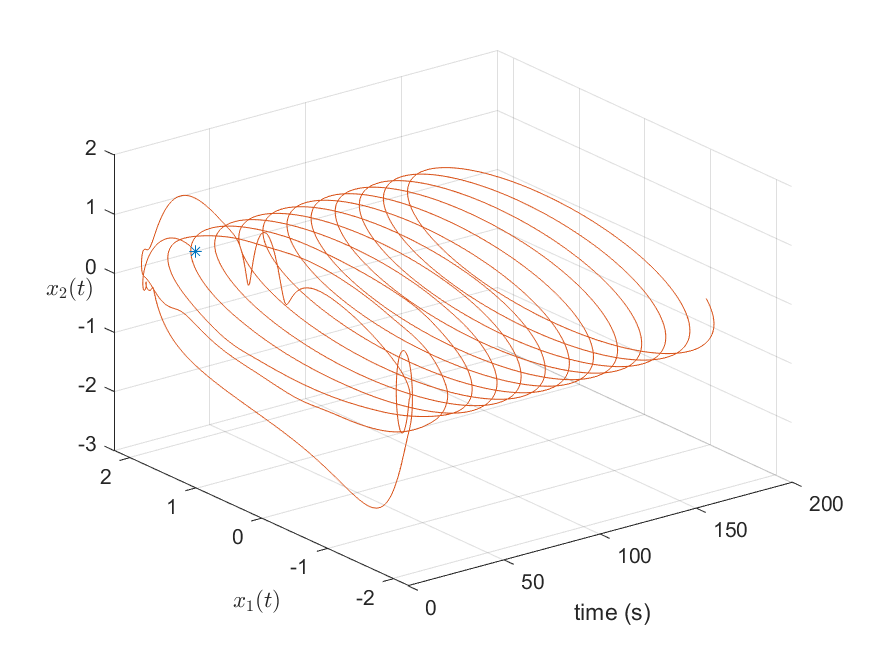,width=\textwidth}
\caption{State trajectory of the Duffing system (*: initial point)}.\label{fig_duf_1}
\end{figure}
\begin{figure}%[htbp]
\centering
\epsfig{figure=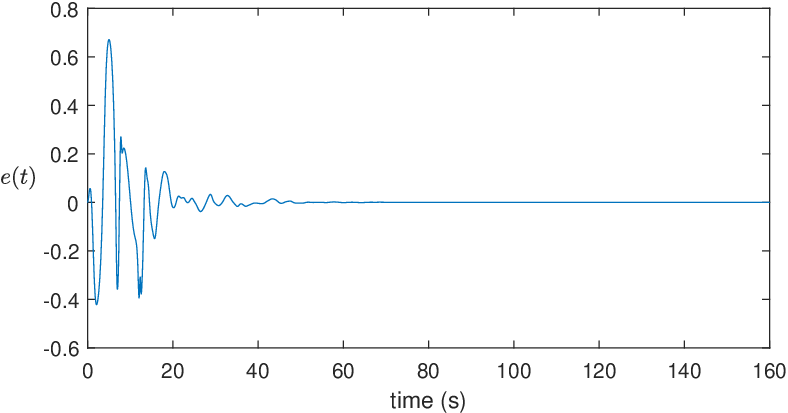,width=\textwidth}

\caption{ Time profile of the tracking error $e=y-h(v,w)$ for the Duffing system with $h(v,w)=v_1$.}\label{fig_duf_2}
\end{figure}
%\end{figure}
\begin{figure}%[htbp]
\centering
\epsfig{figure=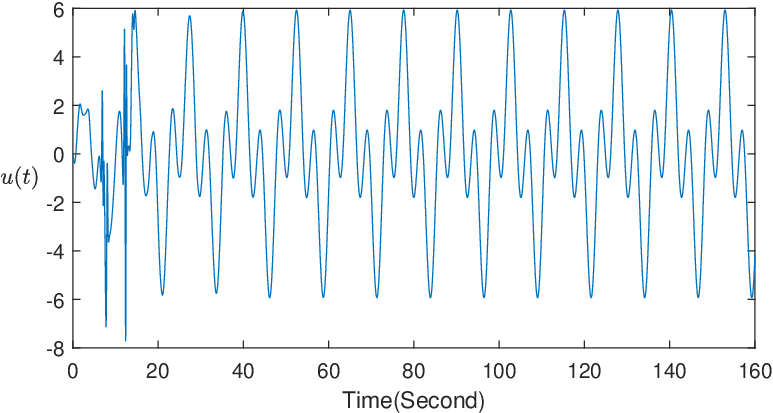,width=\textwidth}

\caption{Time profile of the control input for the Duffing system.}\label{fig_duf_input}\quad
\end{figure}
\begin{figure}%[htbp]
\centering
\epsfig{figure=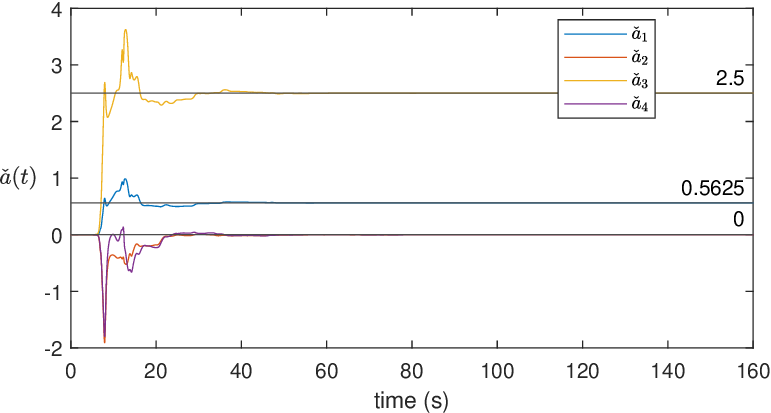,width=\textwidth}

\caption{Estimated parameters of the steady-state dynamics for the Duffing system.}\label{fig_duf_4}\quad
\end{figure}
\begin{figure}[htbp]
\centering
\epsfig{figure=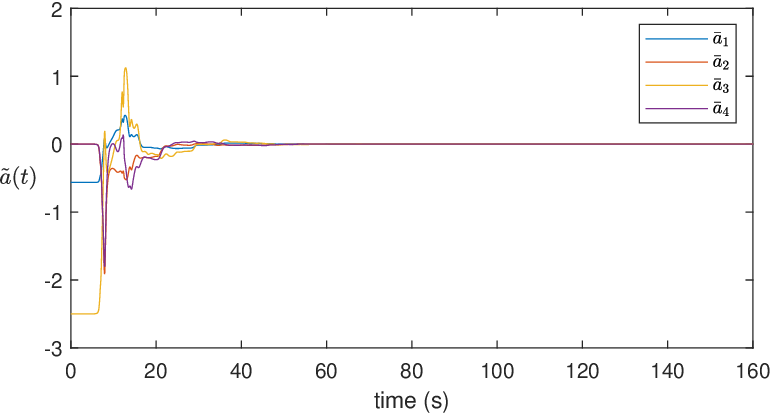,width=\textwidth}

\caption{Parameter estimation error of the steady-state dynamics for the Duffing system.}\label{fig_duf_3}
\end{figure}
Consider the nonlinear system modelled by a controlled Duffing system \cite{liu2008asymptotic}:
\begin{subequations}\label{ex:duffin}
    \begin{align}
    \dot{x}_1=&\ x_2,  \\
    \dot{x}_2=&-c_1x_1-c_2x_1^3-c_3 x_2+u+d(t),
    \end{align}
\end{subequations}
where $\col(x_1,x_2)\in \mathds{R}^2$ is the state; $c_1=1.5$, $c_2=-2$, and $c_3=0.5$ are the coefficients; and the external disturbance is $d(t)=A\cos(\sigma t + \psi)$ with unknown amplitude, frequency, and phase, which can be generated by an uncertain exosystem in the form \eqref{eqn: exosystem system} with 
\begin{align}\label{ex:exosystem}
S(\sigma)=\left[\begin{matrix}0 &\sigma \\ -\sigma & 0\end{matrix}\right],\quad v=\left[\begin{matrix}v_1\\ v_2\end{matrix}\right],\quad e = y - h(v,w), 
 \end{align}
where $h(v,w)=v_1$, $\sigma\in \mathds{S}=\{\sigma\in \mathds{R}: \sigma \in [0.1,1]\}$ and $ \mathds{V}=\{v\in \mathds{R}^2: \|v\|\leq 2.1\}$ with $\sigma$ being the unknown constant $0.5$. 

Under the stated assumptions, \cite{liu2009parameter} showed that there exists a solution of $\bm{\hat{x}}_2(v,\sigma,w)$, polynomial in $v$, satisfying
\begin{align*}
\frac{d^{4}\bm{\hat{x}}_2}{dt^{4}}+a_1 \bm{\hat{x}}_2+a_2 \frac{d\bm{\hat{x}}_2}{dt}+a_3 \frac{d^{2}\bm{\hat{x}}_2}{dt^{2}}+a_4\frac{d^{3}\bm{\hat{x}}_2}{dt^{3}}=0,
\end{align*}
with unknown true value vector $$a\equiv\underbrace{\col(9\sigma^4, 0,10 \sigma^2,0)}_{\col(a_1,a_2,a_3,a_4)}$$ in \eqref{stagerator}. For the control law \eqref{adESC-1b}, we can choose $\rho(e)=2+2e^2$ based on \eqref{rho-inquality-1}  to make the the inequality \eqref{U1-derivative-r} to be negative definite, $\lambda_1=4$ and $\lambda_2=4$ are chosen to make the matrix $A=A_c-\lambda C_c$ in \eqref{input-filter} to be Hurwitz, $m_1=1$, $m_2=5.1503$, $m_3=13.301$,
$m_4=22.2016$, $m_5=25.7518$, $m_6=21.6013$, $m_7=12.8005$, and $m_8=5.2001$ are chosen to make $M$ in \eqref{explicit-mas1} defined in \eqref{MNINter} to be Hurwitz. The simulation starts with the initial conditions 
$x(0)=\col(1,1)$, $v(0)=\col(1,2)$, $\hat{x}(0)=\textbf{0}_2$, $\eta(0)=\textbf{0}_8$, and $\hat{k}(0)=0$. %Cite the explicit theoretical results used for selection of the design parameters for each case study.
The control law stabilizes the system, with the tracking error converging to nearly zero within 50 seconds 
(Figs.\ 
\ref{fig_duf_1}--\ref{fig_duf_2}), and the control signal converging to a periodic signal (Fig.\ \ref{fig_duf_input}).
The estimated parameters of the steady-state dynamics for the closed-loop Duffing system converge to zero within the same time period, which further implies the the estimated parameters converge to the true values  (Fig.\ \ref{fig_duf_4}--\ref{fig_duf_3}).
%Finally, Fig. \ref{figessec} shows the trajectory of suspension deflection without or with active suspension. The control signal,  $u(t)$, is active for $t\geq 40$. 

%\subsection{Example 3: Application in Continuous stirred tank reactor (CSTR)}\label{numerexam2}

%Consider a continuous two-state bioreactor mode borrowed from \cite{inguva2023dynamics} described by the following equations:

%\begin{align}\label{cstr-non}
%   \dot{x}_1=&-Dx_1+\frac{\mu_{\max }x_2}{K_s+x_2 } x_1,\nonumber\\
 %   \dot{x}_2=&-Dx_2 -\frac{\mu_{\max }x_2}{(K_s+x_2)(v_{\min}+ax_2)^{\rho} }x_1+Du,\\
 %   e=&x_2-y_r,\nonumber
%\end{align}

\section{Conclusion}\label{conlu}
This article proposes a nonadaptive nonlinear robust output regulation approach for general nonlinear output feedback systems with error output. 
The proposed nonadaptive framework transforms the robust output regulation problem into a robust non-adaptive stabilization method that is effective for systems with Input-to-State Stable dynamics. 
%
%The integration of a nonparametric learning framework ensures the viability of the nonlinear mapping and eliminates the need for specific Lyapunov function construction {\myr and the commonly employed parameterized assumption on the nonlinear
%system}. 
%
The approach is illustrated in one numerical example, involving a controlled Duffing system, showing convergence of the parameter estimation error and of the tracking error to zero.
%
% {\myr Future research will be conducted on applying the proposed nonparametric learning framework to control the impinging jet mixer for improving the productivity of the solid lipid nanoparticles.
% }

\bibliographystyle{siamplain}
\bibliography{myref}
\end{sloppypar}
\end{document}